\documentclass[letterpaper, 10 pt, conference]{ieeeconf}  

\IEEEoverridecommandlockouts                              
\usepackage{amsmath,amssymb,amsfonts,bm}
\usepackage{algorithm,algorithmic}
\usepackage{graphicx}
\usepackage{textcomp}
\usepackage[utf8]{inputenc}
\usepackage{cleveref}
\usepackage{graphicx}
\usepackage{cite}
\usepackage{enumerate}
\usepackage{array}
\usepackage{makecell}
\usepackage{cellspace}
\usepackage{braket}
\usepackage{mathtools}
\usepackage{blindtext}
\usepackage{etoolbox}
\usepackage[outdir=./]{epstopdf}
\usepackage{fp, tikz, pgfplots}
\usepackage{cases}
\usetikzlibrary{arrows}
\usetikzlibrary{shapes}
\usetikzlibrary{patterns}
\usetikzlibrary{decorations.pathmorphing}
\usetikzlibrary{decorations.pathreplacing}
\usetikzlibrary{decorations.shapes}
\usetikzlibrary{decorations.text}
\usetikzlibrary{shapes.misc}
\usetikzlibrary{decorations.markings}
\usetikzlibrary{arrows.meta}
\usetikzlibrary{backgrounds}
\usepackage[font=small,skip=3pt]{caption}
\usepgfplotslibrary{external}
\usepgfplotslibrary{fillbetween}

\pgfplotsset{width=9\columnwidth /10, compat = 1.13, 
	height = 45\columnwidth /100, grid= major, 
	legend cell align = left, ticklabel style = {font=\scriptsize},
	every axis label/.append style={font=\small},
	legend style = {font=\scriptsize},title style={yshift=0pt, font = \small} }

\tikzset{cross/.style={cross out, draw=black, minimum size=10*(#1-\pgflinewidth), inner sep=0pt, outer sep=0pt},cross/.default={1pt}}

\tikzset{
	myarrow/.style={-{Triangle[length=2mm,width=2mm]}}
}

\crefformat{equation}{(#2#1#3)}
{
	\newtheorem{assumption}{Assumption}
	\newtheorem{lemma}{Lemma}
	\newtheorem{theorem}{Theorem}

	\newtheorem{remark}{Remark}
}

\Crefname{assumption}{Assumption}{Assumptions}
\Crefname{lemma}{Lemma}{Lemmata}
\Crefname{theorem}{Theorem}{Theorems}
\Crefname{Definition}{Definition}{Definitions}
\Crefname{Proposition}{Proposition}{Propositions}
\Crefname{Corollary}{Corollary}{Corollaries}
\Crefname{Algorithm}{Algorithm}{Algorithms}
\crefname{section}{Sec.}{Sections}
\Crefname{section}{Section}{Sections}
\Crefname{algorithm}{Algorithm}{Algorithms}
\Crefname{Remark}{Remark}{Remarks}

\makeatletter
\DeclareRobustCommand{\qed}{%
	\ifmmode 
	\else \leavevmode\unskip\penalty9999 \hbox{}\nobreak\hfill
	\fi
	\quad\hbox{\qedsymbol}}
\newcommand{\openbox}{\leavevmode
	\hbox to.77778em{%
		\hfil\vrule
		\vbox to.675em{\hrule width.6em\vfil\hrule}%
		\vrule\hfil}}
\newcommand{\qedsymbol}{\openbox}
\newenvironment{proof}[1][\proofname]{\par
	\normalfont
	\topsep6\p@\@plus6\p@ \trivlist
	\item[\hskip\labelsep\itshape
	#1:]\ignorespaces
}{%
	\qed\endtrivlist
}
\newcommand{\proofname}{Proof}
\makeatother

\newcommand\xaug{{\tilde{\bm x}}}

\newcommand{\x}{{\bm x}}
\newcommand{\Ndata}{N}
\newcommand{\DSamp}{\mathcal{S}}
\newcommand{\D}{\mathcal{D}} 
\newcommand{\StatSpAug}{{\tilde{\StatSp}}}
\newcommand{\y}{\bm y}

\newcommand{\gsc}{g}
\newcommand{\kernel}{k}
\newcommand{\NSubsumption}{\tilde{N}}
\newcommand{\X}{{\mathbb{X}}} 

\newcommand{\Xaug}{\tilde{\mathcal{X}}}
\newcommand{\XDaug}{\tilde{\mathcal{X}}_{\Ndata}}

\newcommand{\Uin}{\mathbb{U}} 
\newcommand{\K}{{\bm{K}}}
\newcommand{\f}{{\bm f}}
\newcommand{\fsc}{f}
\newcommand{\g}{{\bm{g}}}

\newcommand{\hsc}{h}
\newcommand{\h}{{\bm{h}}}
\newcommand{\w}{{\bm{w}}} 
\newcommand{\wsc}{{w}} 

\newcommand{\transp}{^{\text{T}}}

\newcommand{\Id}{\bm I}

\newcommand{\Nhor}{H}
\newcommand{\Ncons}{S} 
\newcommand{\dimx}{{d_x}} 
\newcommand{\dimu}{{d_u}}
\newcommand{\dimxaug}{d_{\tilde{x}}}
\newcommand{\uin}{{\bm u}}
\newcommand{\uinsc}{u}

\newcommand{\xsamp}{\bm{s}}
\newcommand{\xaugsamp}{\tilde{\bm{s}}}
\newcommand{\xaugsc}{\tilde{x}}
\newcommand{\StatSp}{\mathbb{X}} 
\newcommand{\postmean}{\bm{\mu}}
\newcommand{\postmeansc}{\mu}
\newcommand{\step}{t}

\newcommand{\Nsamps}{{M}} 
\newcommand{\C}{C}
\newcommand{\DSubsumption}{\mathcal{Z}}

\newcommand{\Nlaws}{L}

\DeclarePairedDelimiterX\PBasics[1](){ #1}
\DeclarePairedDelimiterX\EBasics[1][]{ #1}

\title{\LARGE \bf
Data selection for multi-task learning under dynamic constraints
}

\author{Alexandre Capone, Armin Lederer, Jonas Umlauft and Sandra Hirche
\thanks{All authors are with the Department of Electrical and Computer Engineering, Technical University of Munich, 80333 Munich, Germany
        {\tt\small [alexandre.capone, armin.lederer, jonas.umlauft, hirche]@tum.de}}%
}

\begin{document}

\maketitle
\thispagestyle{empty}
\pagestyle{empty}

\begin{abstract}

Learning-based techniques are increasingly effective at controlling complex 
systems using data-driven models. However, most work done so far has focused on 
learning individual tasks or control laws. Hence, it is still a largely 
unaddressed research question how multiple tasks can be learned efficiently and simultaneously on the same system. In 
particular, no efficient state space exploration schemes have been designed for 
multi-task control settings. Using this research gap as our main motivation, we 
present an algorithm that approximates the smallest data set that 
needs to be collected in order to achieve high control performance for multiple 
learning-based control laws. We describe system uncertainty using a 
probabilistic Gaussian process model, which allows us to quantify the impact of 
potentially collected data on each learning-based controller. We then determine 
the optimal measurement locations by solving a stochastic optimization problem 
approximately. We show that, under reasonable assumptions, the 
approximate solution converges towards that of the exact problem. Additionally, 
we provide a numerical illustration of the proposed algorithm.

\end{abstract}


\section{Introduction}

The success of data-driven techniques in control crucially depends on the quality of the available training data 
set~\cite{umlauft2020smart,beckers2017stable,deisenroth2015gaussian}. In reinforcement learning, this difficulty is tackled 
through task-oriented exploration, i.e., by collecting data that is particularly 
useful for the given task~\cite{deisenroth2015gaussian}. However, if the 
task changes, e.g., the system is required to follow a different reference 
trajectory, then the available data might be unsuited to train the corresponding control law, and a new 
exploration phase is necessary. This type of scenario is addressed by multi-task 
reinforcement learning approaches, where policies are sequentially trained for different tasks in order to achieve good overall performance~\cite{wilson2007multi}. However, multi-task reinforcement learning approaches often do not consider constraint requirements \cite{6907421,hessel2019multi,teh2017distral}. Furthermore, if all task-related exploration requirements are amalgamated into a single exploration phase, then the number of system interactions required to obtain good control performance across all tasks is potentially reduced. This is generally desirable, 
as system interactions are often considered costly~\cite{deisenroth2011learning}.

Most techniques for system exploration aim to steer the state to regions that correspond to high system uncertainty \cite{NIPS2018_7538,koller2018learning}, i.e., they aim to achieve a globally accurate model. However, this is intractable for large state spaces, as it implies prohibitively long exploration periods. Moreover, some regions of the state 
space do not 
need to be explored in order to obtain good control performance. Hence, these approaches are not suited to efficiently collect data for multi-task reinforcement learning.

Efficiently exploring the state space of a system to gather data for multiple different tasks poses a twofold challenge. Firstly, the optimal set of hypothetical system measurements needs to be determined. Secondly, an efficient exploration trajectory needs to be determined. In this work, we address this dilemma by proposing an algorithm that approximates the \textit{minimal} number of hypothetical measurement points required to achieve good control performance in several different tasks. This is the main contribution of our paper.
We employ a probabilistic Gaussian process model to quantify model uncertainty, and measure control performance by computing the probability of constraint violation given dynamic constraints. Our algorithm employs a 
random sampling-based approximation, which we show to be exact as the number of samples tend to infinity.

This paper is structured as follows: After a formal problem definition in 
\cref{sec:problem}, the considered Bayesian model is introduced, in 
\cref{sec:GPs}. \Cref{sec:opt} presents the 
algorithm for approximating the optimal measurement locations, which is the main contribution of our paper. A numerical illustration, 
in \cref{sec:simulation}, is followed by a conclusion, in \cref{sec:conclusion}.

\section{Problem Statement}
\label{sec:problem}
We consider a stochastic nonlinear system of the form \footnote{Let $\mathbb{N}$ denote the positive integers, $\mathbb{N}_0\coloneqq \mathbb{N}\cup \{0\}$ the non-negative integers, $\mathbb{R}$ the real numbers, and $\mathbb{R}_{-}$ the negative real numbers,. $\mathcal{P}(\cdot)$ denotes the power set operator. We employ bold notation to denote vectors and matrices and $\leq$ to denote component-wise inequality. Given matrices $\bm{A}\in \mathbb{R}^{m\times n}$, $\bm{B}\in \mathbb{R}^{m\times l}$, $m,n, l\in \mathbb{N}$, we employ brackets accompanied by subscripts $[\bm{A}]_{ij}$ to denote the entry in the $i$-th row and $j$-th column of $\bm{A}$, and brackets without subscripts $[\bm{A} \bm{B}]$ to denote the matrix concatenation of $\bm{A}$ and $\bm{B}$. $\lceil \cdot \rceil$ denotes the ceiling operator, and $\Id_{n}$, $n\in \mathbb{N}$ denotes the $n$-dimensional identity matrix.}
\begin{align}
	\label{eq:SystemDynamics}	
	\begin{split}
		{\x_{\step+1}}  &= \f(\x_{\step},\uin_{\step}) + \g(\x_{\step},\uin_{\step}) + \w_{\step} \\
		& \coloneqq \f(\xaug_{\step}) + \g(\xaug_{\step})+ \w_{\step}
\end{split}	\end{align}
where ${\x_{\step} \in \StatSp \subseteq \mathbb{R}^{\dimx}}$, ${\mathbf{u}_{\step} \in \Uin \subseteq \mathbb{R}^{\dimu}}$ are the system's states, control inputs at time step $\step \in \mathbb{N}$, respectively. The system is perturbed by normally distributed process noise $\w_{\step} \sim \mathcal{N}(\bm{0},\bm{I}_{\dimx})$. The vector ${\xaug_{\step}:= (\x_{\step}, \uin_{\step}) \in \StatSpAug}$, where ${\StatSpAug:=\StatSp \times \Uin}$, concatenates the state $\x_{\step}$ and the control inputs $\uin_{\step} $, and is introduced for the sake of brevity. The function $\f :\StatSpAug \mapsto \StatSp$ is known a priori, whereas $\g : \StatSpAug \mapsto \StatSp$, is an unknown function, for which we assume to have a probabilistic model, as discussed in \Cref{sec:GPs}.
\begin{remark}
Assuming that $\f(\cdot)$ is known does not constitute a restrictive requirement, as it encompasses the scenario without any precise prior system knowledge, i.e., $\f(\xaug_{\step})= \x_{\step}$.
\end{remark}We assume that we are given $\Nlaws \in \mathbb{N}$ data-driven control laws $\uin^j: \X \times \mathcal{P}( \StatSpAug \times \StatSp) \times \mathbb{N} \mapsto \X$. The second argument of the control laws corresponds to a set of $\Ndata$ measurement data pairs 
\begin{align}
\D_{\Ndata} \coloneqq \left\{ \xaug^{(i)}, \f(\xaug^{(i)}) + \g(\xaug^{(i)}) + \w^{(i)}\right\}_{i \in \mathbb{N}_{\leq \Ndata}},
\end{align}
which is to be collected, e.g., via system exploration. The third argument of $\uin^j(\cdot,\cdot,\cdot)$ is the time step $\step$, which accounts for any time-dependent component of the control laws, e.g., time-varying reference trajectories. This type of control law is frequently employed in learning-based settings \cite{Capone2019BacksteppingFP,koller2018learning}. For the sake of notational simplicity, we henceforth use $\XDaug \coloneqq \left\{ \xaug^{(i)}\right\}_{i \in \mathbb{N}_{\leq \Ndata}} \in \StatSpAug^{\Ndata}$ to denote the locations of $\Ndata$ system measurements. Furthermore, make the following assumption.
%

\begin{assumption}
	\label{assumption:uisrealanalytic}
	The control laws $\uin^j(\cdot,\cdot,\cdot)$ are real analytic with respect to the first argument..
\end{assumption} 

In particular, this implies that the control laws $\uin^j(\cdot,\cdot,\cdot)$ are smooth with respect to the state. This applies for many commonly used control laws, e.g., PID-controllers and neural networks with smooth activation functions.

Each control law $\uin^j(\cdot,\cdot, \cdot)$ is required to fulfill a different task, which is expressed as a series of constraints
\begin{align}
	\label{eq:Constraints}
	\h^{j}_{\step}\left(\xaug^j_{\step}\right) \leq \bm{0}, \quad \forall \ \step\in \mathbb{N}_{\leq \Nhor}, \ j\in \mathbb{N}_{\leq \Nlaws}
\end{align}
over a finite time horizon of $\Nhor$ steps. Here $\h^j_{\step}: \StatSpAug \mapsto \mathbb{R}^{\Ncons}$ are nonlinear constraint functions, $\Ncons \in \mathbb{N}$ denotes the number of constraints corresponding to the $j$-th control law, $\xaug^j_{\step}\coloneqq (\x_\step,\uin^j(\x_{\step}, \D))$. Such constraints are often linear, e.g., in the case of energy or input saturation constraints, or polynomial, e.g., in the case of tracking error performance requirements. In this work, we consider the following, more general case:
\begin{assumption}
	\label{ass:hisrealanalytic}
	The entries $[\h^{j}_{\step}(\cdot)]_i$ of the functions $\h^{j}_{\step}(\cdot)$ are non-constant and real analytic \cite{krantz2002primer}.
\end{assumption}

\begin{remark}
The proposed method extends straightforwardly to the more general case where both the horizon $\Nhor$ and number of constraints $\Ncons$ are different for each control law. However, we do not consider this case, as it would incur cumbersome notation.
\end{remark}

We aim to obtain the smallest possible set of measurement locations $\Xaug^* := \left\{\xaug^{(i),*}\right\}_{i=1,\cdots,\Ndata^*}$, such that the corresponding data set $\D^*$, if collected and used to design the control laws $\uin^j(\cdot,\cdot, \cdot)$, yields system trajectories that satisfy \eqref{eq:Constraints} with high probability, i.e.,
\begin{flalign}
	\label{eq:MinimizationProblem}
	\thickmuskip = 3mu
	\begin{alignedat}{4}
		\thickmuskip = 0mu
	\medmuskip = 0mu
	\thinmuskip = 0mu
	\Xaug^* = &\omit{\rlap{$ \  \arg \min \limits_{\XDaug \in \mathcal{P}(\StatSpAug)} \Ndata $}}&&& \\
	\text{s.t.} \quad &  \D_{\Ndata} &&= \left\{ \xaug^{(i)}, \f(\xaug^{(i)}) + \g(\xaug^{(i)}) + \w^{(i)}\right\}_{i \in \mathbb{N}_{\leq \Ndata}} &\\
	&\omit{\rlap{$\text{P}\Bigg( \h^j_{\step}(\xaug^j_{\step}) \leq \bm{0}, \quad  \forall \ \step\in \mathbb{N}_{\leq \Nhor}, \ j\in \mathbb{N}_{\leq \Nlaws} \Bigg) > 1- \delta,$}}& 
	\end{alignedat}&&
\end{flalign}
where  $0 < \delta<1$ is a predetermined scalar that specifies the desired probability of constraint violation. The probability operator $\text{P}(\cdot)$ describes the probability of an event given process noise $\w_{\step}$ and the a priori distribution that we assume for the unknown function $\g(\cdot)$, as discussed in \Cref{sec:GPs}.

\begin{remark}
	Since the system dynamics are unknown , the measurements in an arbitrary data set $ \D_{\Ndata}$ are  hypothetical. However, by assuming a priori distribution over $\g(\cdot)$, we are able to determine the impact of measurement locations $\XDaug$ on control performance. 
\end{remark}

\begin{remark}
	In order to guarantee convergence of the method proposed in this paper, we require a solution $\Xaug^*$ of \eqref{eq:MinimizationProblem} to satisfy the chance constraints strictly. However, this is not a severe restriction, as $\delta$ is a design choice.
\end{remark}

Finding an optimal set $\Xaug^*$ under uncertainty is generally impossible without considering further assumptions. Hence, we restrict ourselves to the case where the controllers are specified in a way that the desired closed-loop behavior is achievable:
\begin{assumption}
	\label{ass:Problemhassolution}
	The optimization problem \eqref{eq:MinimizationProblem} is feasible for a finite $\Xaug^*$, i.e., $\rvert \Xaug^* \rvert \eqqcolon \Ndata^* < \infty$.
\end{assumption}

Furthermore, we assume that the optimal data set is contained within a known compact subset of $\StatSp$:
\begin{assumption}
	\label{ass:OptimuminCompactSet}
	There exists a known compact subset $\StatSpAug^{*} \subset \StatSpAug$, such that $\xaug^{(i),*} \in \StatSpAug^{*}$ for all $i\in \left\{1,\cdots,\Ndata^*\right\}$.
\end{assumption}
This does not constitute a very restrictive assumption, since control tasks are typically restricted to a compact subset of the state space, which in turn implies that only system information within a compact subset is required to achieve good control performance.

In order to streamline notation, we henceforth subsume measurement data and system trajectories of \eqref{eq:SystemDynamics} as
\begin{align}
	\DSubsumption_{\NSubsumption} \coloneqq \left\{ \tilde{\bm{z}}_n, \f(\tilde{\bm{z}}_n) + \g(\tilde{\bm{z}}_n) + \w_n \right\}_{n\in \mathbb{N}_{\leq \NSubsumption}}
\end{align}
where $\NSubsumption \coloneqq \Ndata+\Nlaws(\Nhor+1)$,
\begin{align}
	\tilde{\bm{z}}_n = \begin{cases}
\xaug^{(d_n)}, & n =0\cdots,\Ndata-1 \\
\xaug^{j_n}_{\step_n} & n = \Ndata,\cdots, \NSubsumption 
\end{cases}
\end{align}
$d_n\coloneqq n+1 $ $j_n \coloneqq \lceil (n-\Ndata)/(\Nhor+1)\rceil$, $\step_n\coloneqq n - \Ndata -(j(n)-1)(\Nhor+1)$. 
\section{Probabilistic model}
\label{sec:GPs}
In order to quantify the uncertainty corresponding to the unknown component $\g(\cdot)$, we assume a GP distribution over $\g(\cdot)$. Formally, a GP is a collection of random variables, of which any finite subset is jointly normally distributed \cite{Rasmussen2006}. In order to assess how data collected in the future will potentially affect control performance, we need to quantify how model uncertainty decreases as new data points are added. To this end, we consider \textit{hypothetical} data sets $\DSamp_{\NSubsumption} = \left\{ \xaug_n, \f(\xaug_n) + \g^s(\xaug_n) + \w_n \right\}_{n \in \mathbb{N}_{\leq \NSubsumption}}$, 
which are sampled from the GP distribution. Here we employ the superscript $s$ to emphasize that $\g^s(\cdot)$ is a \textit{sample} function evaluation, as opposed to an evaluation of the true function $\g(\cdot)$. This is explained in detail in the sequel.

\begin{remark}
A GP model can be trained using measurement data from the true system \eqref{eq:SystemDynamics}. For the sake of notational simplicity, we analyze the setting where no prior measurement data from the \textit{true system} is available, and show exclusively how to draw samples from a GP in a recursive fashion. However, this does not constitute a loss of generality, a posterior GP distribution after training satisfies the requirements used in this paper \cite{Rasmussen2006}.
\end{remark}

We begin by introducing GPs for the case where $\dimx = 1$, and then describe how they can be generalized to a multivariate setting. A GP is fully specified by a mean function, which we set to zero without loss of generality \cite{Rasmussen2006}, and a positive definite kernel $k :\mathbb{R}^{\dimxaug} \times \mathbb{R}^{\dimxaug}\mapsto \mathbb{R}$. Given a sample data set $\DSamp_{n}$, a subsequent sample evaluation at an arbitrary augmented state $\xaug$ is normally distributed, i.e., $\gsc^s(\xaug) \sim \mathcal{N}(\mu_{n+1}(\xaug), \sigma^2_{n+1}(\xaug))$, with mean and variance given by
\begin{align}
\label{eq:GPMean}
\begin{split}
\postmeansc_{n+1} \left(\xaug\right) \coloneqq &\postmeansc \left(\xaug \vert \DSamp_{n} \right) 
= \bm{\kernel}_n\transp\left(\xaug\right) \K_n^{-1} \y_n
\end{split}\\
\medmuskip=0mu
\thinmuskip=0mu
\label{eq:GPVar}
\begin{split}	
\sigma_{n+1}^2\left(\xaug \right) \coloneqq & \sigma^2\left(\xaug \vert \DSamp_{n} \right) =\kernel\left(\xaug,\xaug\right) - \bm{\kernel}_n\transp\left(\xaug\right)\K_n^{-1} \bm{\kernel}_n\left(\xaug\right),
\end{split} 
\end{align}
where $\bm{\kernel}_n(\xaug) = (\kernel(\xaug,\xaugsamp_1), \cdots, \kernel(\xaug,\xaugsamp_n))$, $\y_n = (\gsc^s(\xaugsamp_1),\cdots, \gsc^s(\xaugsamp_n))$, and the entries of the covariance matrix are given by $
[\K_{n}]_{ij} =\kernel(\xaugsamp_i,\xaugsamp_j)$.

Using \eqref{eq:GPMean} and \eqref{eq:GPVar}, we are able to sample measurement data sets as well as system trajectories from the prior GP distribution using
\begin{align}
	\gsc^s(\xaug) \coloneqq \postmeansc_{n+1} \left(\xaug\right) + \sigma_{n+1} \left(\xaug\right)\zeta
\end{align}
and $\zeta \sim \mathcal{N}(0,1)$
In settings where $\dimx > 1$, we model each dimension using a separate GP, i.e., $\g^s(\xaug) \sim \mathcal{N} \left(\postmean_{n}(\xaug) ,\bm{\sigma}^2_{n}(\xaug) \right)$, where
$
\postmean_{n}(\xaug)\coloneqq  \left(\postmeansc(\xaug\vert \DSamp_{1,n}), \ \cdots, \ \postmeansc(\xaug\vert  \DSamp_{\dimx,n}) \right) $, $
\bm{\sigma}^2_{n}(\xaug)
\coloneqq  \text{diag}\Big({\sigma}^2(\xaug\vert \DSamp_{1,n}), \ \cdots, {\sigma}^2(\xaug\vert \DSamp_{\dimx,n})\Big)$, 
and the measurement data and samples are separated for each dimension $d \in\left\{1,\cdots, \dimx\right\}$ as $\DSamp_{d,n} = \left\{ \xaug^{(i)}, \fsc_d(\xaug^{(i)}) + [\g^s(\xaug^{(i)})]_d +  \wsc_d\right\}_{i=1, \cdots, n }$. This approach corresponds to conditionally independent state transition function entries, which is a common assumption for multivariate systems \cite{deisenroth2015gaussian}. 

In the following, we formally state the GP-related assumption required to carry out our method.
\begin{assumption}
	\label{ass:EntriesfromgareGPsamples}
	The entries of $\g(\cdot)$ correspond to samples from a GP with zero mean and known analytic kernel $\kernel(\cdot,\cdot)$, i.e., $[\g]_d(\cdot) \sim \mathcal{GP}(0,\kernel)$ holds for $d=1, \cdots, \dimx$.
\end{assumption}

In particular, Assumption \ref{ass:EntriesfromgareGPsamples} implies that the expected value of an arbitrary state $\xaug_{\step}^j$ at time $\step$ under control law $j$ is given by
\begin{align}
\label{eq:RewrittenEstDynamics}
\begin{split}
\text{E}_{\g,\w}\left(\x^j_{\step}\right) =& \int \limits_{\StatSp^{2t}} \xsamp_{n_{j,\step}}\prod \limits_{i=0}^{n_{j,\step}}  \text{p}(\bm{\zeta}_i)  d\bm{\zeta}_{i},
\end{split}
\end{align}
where $n_{j,\step}\coloneqq \Ndata + (j-1)(\Nhor+1) + \step $, $\text{E}_{\g,\w}(\cdot) $ denotes the expected value with respect to the unknown function $\g(\cdot)$ and the process noise, and the samples are computed recursively using
\begin{align*}
\xsamp_{n+1} =& \f(\xaugsamp_n) +\bm{\mu}_{n}(\xaugsamp_n) + \left[\bm{\sigma}_{n}(\xaugsamp_n)\quad  \bm{Q}\right]\bm{\zeta}_n, \\
\xsamp_{n} =& \left( \xaugsamp_n, \uin^j(\xaugsamp_n,\DSamp_{\Ndata},\step)\right) \quad \forall \ n_{j,0} \leq n < n_{j,\step} \\
	\begin{split}
	\DSamp_{i} =& \Big\{ \xaugsamp_n, \f(\xaugsamp_n) +\bm{\mu}_{n}(\xaugsamp_n) + \left[\bm{\sigma}_{n}(\xaugsamp_n)\quad  \bm{Q}\right]\bm{\zeta}_n \Big\}_{n \in \mathbb{N}_{\leq i}} 
	\end{split}
\end{align*}
Here $\text{p}(\bm{\zeta}_n) =\mathcal{N}(\bm{0},\bm{I}_{2\dimx})$. Note that we require the random variables $\bm{\zeta}_i$ to have dimension $2 \dimx$ in order for the GP samples 
\begin{align}
\label{eq:GPSample}
\g^s(\xaugsamp_n) = \bm{\mu}_{i}(\xaugsamp_n) + \bm{\sigma}_{i}(\xaugsamp_n)\left[\bm{\zeta}_n\right]_{1:\dimx},
\end{align}
where $[\bm{\zeta}_n]_{1:\dimx}$ denotes the first $\dimx$ entries of $\bm{\zeta}_i$, to be uniquely defined \cite{Rasmussen2006}. 

Since our goal is to find the smallest possible set of measurement points $\Xaug^*$, it is reasonable to assume that $\Xaug^*$ does not contain any measurement locations that provide identical information. In terms of a GP distribution, this is expressed as follows:
\begin{assumption}
	\label{ass:Minimizerisinformative}
	Let $\Xaug^*$ be the minimizer of \eqref{eq:MinimizationProblem}. Then $\bm{\sigma}_{n}(\xaug^{(n+1),*}) \neq 0$ holds for ${n \in \mathbb{N}_{\leq\Ndata^*-1}}$.
\end{assumption}

For many commonly used kernels, e.g., squared exponential kernels, \Cref{ass:Minimizerisinformative} implies that $\Xaug^*$ does not contain identical measurement locations. 

%

\section{Two stage optimization}
\label{sec:opt}
We now describe the optimization scheme used to approximate the optimal solution of \Cref{eq:MinimizationProblem}, and provide a corresponding theoretical analysis.

Since each control law $\uin^j(\cdot,\cdot,\cdot)$ is fully specified by the training data $\D_{\Ndata}$, the probability distribution of a trajectory obtained using any two different control laws $\uin^j(\cdot,\cdot,\cdot)$, $\uin^i(\cdot,\cdot,\cdot)$, $i\neq j$, are conditionally independent given $\D_{\Ndata}$, i.e., $
	p(\xaug^j_{\tau},\xaug^i_{\step} \vert \D_{\Ndata}) = p(\xaug^j_{\tau} \vert \D_{\Ndata})p(\xaug^i_{\step} \vert \D_{\Ndata})
$
for all $\step,\tau \in \left\{1,\cdots,\Nhor \right\}$. Moreover, since the control laws $\uin^j$ are deterministic given $\D_{\Ndata}$, we have $
p(\xaug^j_{\tau} \vert \D_{\Ndata}) = p(\x^j_{\tau} \vert \D_{\Ndata})$. Hence, similarly to \eqref{eq:RewrittenEstDynamics}, computing the probability of constraint satisfaction for a set of measurement points $\XDaug$ amounts to evaluating the integral
\begin{align}
\label{eq:ProbConstraintSatisfaction}
\thickmuskip=0.5mu
\begin{alignedat}{2}
&\C_{\Ndata}\Big( \XDaug \Big) \coloneqq \text{P}\Bigg( \h^j_{\step}(\xaug^j_{\step}) \leq \bm{0}, \ \forall \ \step \in \mathbb{N}_{ \step \leq \Nhor}, \ j \in \mathbb{N}_{ \step \leq \Nlaws} \Bigg) \\
=& \prod_{n=1}^{\NSubsumption} \ \int\limits_{\X^{\Nhor \Ndata}}  \bm{1}_{\mathbb{R}_{-}^{\dimxaug}}\left(\h^{j_n}_{\step_n}(\xaugsamp_{n_{j,\step}}) \right)   p\left(\xaugsamp_{n_{j,\step}} \Big\vert \DSamp_{n_{j,\step}} \right) d\bm{\zeta}_{n}
\end{alignedat}
\end{align}  $\bm{1}_{\mathbb{R}_{-}^{\dimxaug}}(\cdot)$ is the indicator function of $\mathbb{R}_{-}^{\Ncons}$. 

 Generally, computing \eqref{eq:ProbConstraintSatisfaction} is intractable, which renders a direct approach to solving \eqref{eq:MinimizationProblem} infeasible. In this work, we employ a two-stage optimization approach, which yields an approximation of the optimal solution $\Xaug^*$ with probability $1$. This is achieved by repeatedly defining a fixed number of data points $\Ndata$ and maximizing the Monte Carlo approximation $C_{\Ndata}^{\Nsamps}(\XDaug)$ of \eqref{eq:ProbConstraintSatisfaction}. If the maximal approximate probability of constraint satisfaction is lower than the desired bound $1-\delta$, the number of data points $\Ndata$ is increased and the procedure is repeated. This is detailed in \Cref{alg:DataSelCtrl}.

\begin{algorithm}
	\caption{Data selection for multi-task learning (DS-ML)}
	\label{alg:DataSelCtrl}
\begin{algorithmic}[1]
	\REQUIRE M, $\f(\cdot)$, $\bm{Q}$
	\STATE Set $\Ndata = 0$
	\WHILE {$\C_{\Ndata}^{\Nsamps}(\XDaug^N) \leq 1-\delta$}
	\STATE Set $\Ndata \leftarrow \Ndata+1$
	\STATE {$ \forall \ m \in \mathbb{N}_{\leq \Nsamps}, n \in \mathbb{N}_{\leq \NSubsumption},$ sample ${\bm{\zeta}^{m}_{n} \sim \mathcal{N}(\bm{0},\bm{I}_{2\dimx})}$}
	\STATE \label{algstep:IterativeOpt}
	Solve
	\begin{align*}
		\XDaug^{\Nsamps} \ &\omit{\rlap{$=\arg \max\limits_{\XDaug} \C_{\Ndata}^{\Nsamps}(\XDaug)$}} \\
		& \omit{\rlap{$\coloneqq\arg \max\limits_{\XDaug} \frac{1}{\Nsamps}\sum \limits_{m=1}^{\Nsamps}\prod\limits_{j=1}^{\Nlaws} \prod\limits_{t=1}^{\Nhor} \bm{1}_{\mathbb{R}_{-}^{\Ncons}}\left(\h^j_{\step}(\xaugsamp_{n_{j,\step}}) \right) $ }} \\
		\text{s.t.} \quad 
		&\omit{\rlap{$\forall \ m \in \mathbb{N}_{\leq \Nsamps}, n \in \mathbb{N}_{\leq \NSubsumption}, j \in  \mathbb{N}_{\leq \Nlaws}, \step \in \mathbb{N}_{\leq \Nhor} $}} \\
\xsamp^m_{n+1} =& \f(\xaugsamp^m_n) +\bm{\mu}^m_{n}(\xaugsamp^m_n) + \left[\bm{\sigma}^m_{n}(\xaugsamp^m_n)\quad  \bm{Q}\right]\bm{\zeta}^m_n, \\
\xsamp^m_{n} =& \left( \xsamp^m_n, \uin^j(\xaugsamp^m_n,\DSamp^m_{\Ndata},\step)\right) \quad \forall \ n_{j,0} \leq n < n_{j,\step} \\
\begin{split}
\DSamp^m_{n} =& \Big\{ \xaugsamp^m_i, \f(\xaugsamp^m_i) +\bm{\mu}^m_{i}(\xaugsamp^m_i) \\
&+ \left[\bm{\sigma}^m_{i}(\xaugsamp^m_i)  \ \bm{Q}\right]\bm{\zeta}^m_i \Big\}_{i \in \mathbb{N}_{\leq n}} 
\end{split}
	\end{align*}
	\ENDWHILE
	\STATE Set $\XDaug^{\Nsamps,*} = \XDaug^{\Nsamps}$ 
	\RETURN $\XDaug^{\Nsamps,*}$ 
\end{algorithmic} 
\end{algorithm}
 
\subsection{Theoretical Analysis}
We now derive formal guarantees for the approximate solution $\XDaug^{\Nsamps}$ obtained with \Cref{alg:DataSelCtrl}. 
To this end, we prove some preliminary results.
\begin{lemma}
	\label{lemma:sigmaeq0hasmeasurezero}
	Let \Cref{ass:Minimizerisinformative} hold and let $\DSamp_{n}$ be a sample data set. Furthermore, let $\sigma_{n}^2\left(\cdot \right) $ be the corresponding posterior covariance and let $\uin_{\step}^j(\cdot)$ be a control law that satisfies \Cref{assumption:uisrealanalytic}. Then $\sigma_{n}^2\left(\x, \uin^j(\x) \right) \neq 0$ holds for all $\x \in \StatSp$ up to a set of measure zero.
\end{lemma}
\begin{proof}
	 Non-zero real analytic functions are non-zero almost everywhere, and the concatenation of real analytic functions is also real analytic. These are well-known properties of real-analytic functions \cite{krantz2002primer}. Hence, we only need to show that $\sigma_{n}^2\left(\x, \uin^j(\x) \right)$ is a real-analytic function of $\x$. Since the $\sigma_{n}^2\left(\cdot \right) $ corresponds to a sum of kernel evaluations, it is analytic. As $\uin^j(\cdot,\D_{\Ndata},\step)$ is also analytic, this implies the desired result.
\end{proof}
This enables us to show that the state is on a set of measure zero with probability one.
\begin{lemma}
	\label{lemma:Probxinzeromeasureiszero}
	Let \Cref{ass:EntriesfromgareGPsamples,ass:Minimizerisinformative,ass:hisrealanalytic} be satisfied, and let  $\StatSp_0 \subset \StatSp$ be an arbitrary subset of the state space with measure zero. Then $
	\text{P}( \x^j_{\step} \in \StatSp_0 ) = 0 $ holds.
\end{lemma}
\begin{proof}
	Assume, without loss of generality, that $j=1$. We prove the result by induction for $\Ndata=0$, and then discuss how it extends to an arbitrary $\Ndata \in \mathbb{N}$. The probability that the state lies within an arbitrary set of measure zero at time step $\step$ is then given by
	\begin{align}
	\label{eq:InductionMiddle}
	\begin{split}
	&\text{P}\Bigg( \x^1_{\step} \in \StatSp_{0} \Bigg) 
	= \int \limits_{\StatSp^{2n}} \bm{1}_{\StatSp_{0}}\Big(\f(\xaugsamp_{n-1}) + \bm{\mu}_{n-1}(\xaugsamp_{n-1}) \\
	&\qquad  + \left[\bm{\sigma}_{n-1}(\xaugsamp_{n-1})\quad  \bm{Q}\right]\bm{\zeta}_{n-1}\Big) \prod_{i=0}^{n-1}\text{p}(\bm{\zeta}_{i})  d\bm{\zeta}_{i} 
	\end{split}
	\end{align}
	Since $\f(\xaugsamp_{n-1})$ and $\bm{\mu}_{n-1}(\xaugsamp_{n-1})$ are constant with respect to $\bm{\zeta}_{n-1}$, and the measure of $\StatSp_0$ is translation-invariant, it suffices to show
	\begin{align*}
		&\int \limits_{\StatSp^{2n}} \bm{1}_{\StatSp_{0}}\Big( \left[\bm{\sigma}_{n-1}(\xaugsamp_{n-1})\quad  \bm{Q}\right]\bm{\zeta}_{n-1}\Big) \text{p}(\bm{\zeta}_{n-1})  d\bm{\zeta}_{n-1}  \overset{!}{=} 0,
	\end{align*}
	which we achieve by induction. For $\step = 1$, we have
	\begin{align*}
	\begin{split}
	&\int \limits_{\StatSp^{2}} \bm{1}_{\StatSp_0}\left(\left[\bm{\sigma}_{0}(\xaugsamp^j_0)\quad  \bm{Q}\right]\bm{\zeta}_0\right)\text{p}(\bm{\zeta}_0)  d\bm{\zeta}_{0} 
	= \int \limits_{\StatSp}\Big(\int \limits_{\StatSp^{}} \bm{1}_{\StatSp_0}\left(\x\right) \\
	&\times \text{p}\left(\bm{\sigma}^{-1}_{0}(\xaugsamp_0) \x - \bm{\zeta}_0''\right)  \bm{\sigma}^{-1}_{0}(\xaugsamp^j_0) d\x  \Big)\text{p}(\bm{\zeta}_{0}'') d \bm{\zeta}_{0}'' = 0, 
	\end{split}
	\end{align*} 
	since $\bm{1}_{\StatSp_0}\left(\x\right)=0$ for all $\x \in \StatSp$ up to a set of measure zero. Here we employ the fact that that $\bm{\sigma}_{0}(\xaugsamp_0) = \text{diag}(\kernel(\xaugsamp_0,\xaugsamp_0),\cdots,\kernel(\xaugsamp_0,\xaugsamp_0))$ is invertible for all non-zero kernels, which allows us to
	integrate using the substitution $\x = \bm{\sigma}_{0}(\xaugsamp^j_0)\bm{\zeta}_0'+ \bm{Q}\bm{\zeta}_0''$ and $\bm{\zeta}_0'\coloneqq \left[\bm{\zeta}_i\right]_{1:\dimx}$, $\bm{\zeta}_0''\coloneqq \left[\bm{\zeta}_i\right]_{\dimx+1:2\dimx}$. The expression $\text{p}(\bm{\sigma}^{-1}_{0}(\xaugsamp^j_0) \x - \bm{\zeta}_0'')$ corresponds to a normal distribution with center $\bm{\zeta}_0''$ and scaling matrix $\bm{\sigma}_{0}(\xaugsamp^j_0))^{-1} $, hence it is smooth and integrable with respect to $\x$. 
	Hence, the result holds for $\step = 1$. Note that, due to \Cref{lemma:sigmaeq0hasmeasurezero}, this implies that $\bm{\sigma}_{1}(\xaugsamp_{1})$ is invertible for almost every $\bm{\zeta}_{0}$. Hence, we can assume that $\bm{\sigma}_{n-1}(\xaugsamp_{n-1})$ is invertible for a fixed $n-1$ and almost every $\xaugsamp_{n-1}$. 
and we can apply the same argument as in the case $\step=1$ and obtain the desired result for an arbitrary $\step$ and $j=1$.

 Due to \Cref{ass:Minimizerisinformative}, we can assume that $\bm{\sigma}_{\Ndata}(\cdot)$ is invertible for data sets of size $\Ndata \neq 0$, which enables us to extend the proof to arbitrary $\Ndata$ using the same argument.
\end{proof}
This directly yields the following result:
\begin{lemma}
	\label{corollary:heq0haszeromeasure}
	Let \Cref{ass:EntriesfromgareGPsamples,ass:Minimizerisinformative,ass:hisrealanalytic} be satisfied. Then $
	\text{P}( \h^j_{\step}(\xaug^j_{\step}) = \bm{0}) =  0 $
	holds for all $\step \in \mathbb{N}_{\leq \Nhor}$, $j \in \mathbb{N}_{\Nlaws}$.
\end{lemma}
\begin{proof}
	Since $[\h^j_{\step}(\xaug)]_i$ are real-analytic, $[\h^j_{\step}(\xaug)]_i \neq 0$ holds for all $i \in \mathbb{N}_{\leq \Ncons}$ and all $\xaug \in \StatSpAug$ up to a set of measure zero. 
	By employing \Cref{lemma:Probxinzeromeasureiszero} and the union bound, we obtain
	\begin{align}
	\label{eq:ProbNonSatisfactionForsingleConstraint}
	\text{P}\Bigg( \h^j_{\step}(\xaug^j_{\step}) = \bm{0}\Bigg)  \leq \bigcup_{i\in\mathbb{N}_{\leq\Ncons}} \text{P}\Bigg( \left[\h^j_{\step}(\xaug^j_{\step})\right]_i = {0}\Bigg) =  0.
	\end{align} 
\end{proof}

We now show that the sample average approximations used in Algorithm \ref{algstep:IterativeOpt} converge to the true probabilities of constraint satisfaction \eqref{eq:ProbConstraintSatisfaction}.
\begin{lemma}
	\label{lemma:ConvergenceofCmtoCN}
	Let \Cref{ass:EntriesfromgareGPsamples,ass:Problemhassolution,ass:Minimizerisinformative,ass:hisrealanalytic,assumption:uisrealanalytic} hold, and let $\StatSpAug^{*}$ be given as in \Cref{ass:OptimuminCompactSet}. Then, for an arbitrary $\Ndata \in \mathbb{N}$, the expected value of $\C_{\Ndata}(\cdot)$ is finite valued and continuously differentiable on $(\StatSpAug^{*})^{\Ndata}$, and $\C_{\Ndata}^{\Nsamps}(\cdot)$ converges to $\C_{\Ndata}(\cdot)$ with probability $1$ uniformly in $(\StatSpAug^{*})^{\Ndata}$ as $\Nsamps \rightarrow \infty$. 
\end{lemma}
\begin{remark}
	The proofs of \Cref{lemma:ConvergenceofCmtoCN} and \Cref{theorem:MainResult}, which we state in the following, require Theorem 7.48 and Theorem 5.4 from \cite{shapiro2009lectures}, respectively. Due to space limitations, we do not include them here. However, to facilitate interpretation, we enumerate the technical statements in the proofs of \Cref{lemma:ConvergenceofCmtoCN} and \Cref{theorem:MainResult}, such that they correspond to Theorem 7.48 and Theorem 5.4 from \cite{shapiro2009lectures}.
\end{remark}

\begin{proof}[Proof of \Cref{lemma:ConvergenceofCmtoCN}]
	We show that the the approximation $\C_{\Ndata}^{\Nsamps}(\cdot)$ satisfies all conditions of \cite[Theorem 5.4]{shapiro2009lectures}, enumerated in the sequel as i)-iii), which directly yields the desired result. 
	\begin{enumerate}[i)]
		\item
		Due to \Cref{corollary:heq0haszeromeasure}, the functions $ \bm{1}_{\mathbb{R}_{-}^{\Ncons}}\left(\h^j_{\step}(\xaugsamp^m_{n_{j,\step}}) \right)$ are uniquely defined and continuous for an arbitrary $\step,j\in \mathbb{N}$ and almost every sample  $\bm{\zeta}^m_n$ \cite{shapiro2009lectures}. Hence, $\C_{\Ndata}^{\Nsamps}(\XDaug)$ is continuously differentiable at any $\XDaug \in (\StatSpAug^{*})^{\Ndata}$ for almost every sample $\bm{\zeta}^m_n$.
		\item Since $\C_{\Ndata}^{\Nsamps}(\XDaug) \leq 1$ and $\XDaug \in (\StatSpAug^{*})^{\Ndata}$ is compact, the absolute value of $\C_{\Ndata}^{\Nsamps}(\XDaug)$ is upper bounded by an integrable function on $\XDaug \in (\StatSpAug^{*})^{\Ndata}$.
		\item The samples $\bm{\zeta}^m_n$ are i.i.d.
\end{enumerate}
\end{proof}

\begin{lemma}
	\label{theorem:Shapiro}
	Let \Cref{ass:EntriesfromgareGPsamples,ass:Minimizerisinformative,ass:hisrealanalytic,assumption:uisrealanalytic,ass:Problemhassolution,ass:OptimuminCompactSet} hold. Moreover, let $ \C_{\Ndata}(\cdot)$ be the probability of constraint satisfaction for a data set of size $\Ndata$, let $ \C^{\Nsamps}_{\Ndata}(\cdot)$ correspond to its SAA, and let $\XDaug^{\Nsamps,*}$ denote the output of \Cref{alg:DataSelCtrl}. Then, with probability $1$, for every $\varepsilon \geq 0$, there exists an $\Nsamps_{\varepsilon}$, such that $\C_{\Ndata}(\XDaug^{\Nsamps,*}) - C_{\Ndata}^* \leq \varepsilon$ holds  for all $\Nsamps \geq \Nsamps_{\varepsilon}$.
\end{lemma}
\begin{proof}
	We show that the conditions of \cite[Theorem 5.4]{shapiro2009lectures} are satisfied by $ \C_{\Ndata}(\cdot)$ and $ \C^{\Nsamps}_{\Ndata}(\cdot)$, which yields the desired result. In the following, we employ i)-iv) to enumerate the required conditions, which corresponds to the enumeration in \cite[Theorem 5.4]{shapiro2009lectures}. 
	\begin{enumerate}[i)]
		\item Due to \Cref{ass:OptimuminCompactSet}, $(\StatSpAug^{*})^{\Ndata}$ is non-empty and compact.
		\item Due to \Cref{lemma:ConvergenceofCmtoCN}, $\C_{\Ndata}(\cdot)$ is finite valued and continuously differentiable on $(\StatSpAug^{*})^{\Ndata}$
		\item Due to \Cref{lemma:ConvergenceofCmtoCN}, $\C_M(\cdot)$ converges to $\C_{\Ndata}(\cdot)$ with probability $1$ as $M \rightarrow \infty$, uniformly in $ (\StatSpAug^{*})^{\Ndata}$,
		\item Since we restrict ourselves to the set $(\StatSpAug^{*})^{\Ndata}$, $\XDaug^{\Nsamps,*} \in (\StatSpAug^{*})^{\Ndata}$ holds trivially for all $\Nsamps$.
	\end{enumerate} 
\end{proof}

We now state the main result of this paper, namely that \Cref{alg:DataSelCtrl} is able to approximate an optimal solution arbitrarily accurately using a high enough but finite number of random samples $\Nsamps$.
\begin{theorem}
	\label{theorem:MainResult}
	Let \Cref{ass:EntriesfromgareGPsamples,ass:Minimizerisinformative,ass:hisrealanalytic,assumption:uisrealanalytic,ass:Problemhassolution} hold, and let $\XDaug^{\Nsamps,*}$ denote the output of \Cref{alg:DataSelCtrl}. Then, with probability $1$, for every $\varepsilon > 0$, there exists an $\Nsamps_{\varepsilon}$, such that $\C_{\Ndata}(\XDaug^{\Nsamps,*}) - C^* \leq \varepsilon$ holds  for all $\Nsamps \geq \Nsamps_{\varepsilon}$.
\end{theorem}
\begin{proof}
	The result holds if the approximate optima $\C_{\Ndata}^{\Nsamps}(\XDaug^{\Nsamps})$, $\Ndata = 1, \cdots,\Ndata^* $, obtained in Step \ref{algstep:IterativeOpt} of \Cref{alg:DataSelCtrl} converge uniformly to the true solutions $\C_{\Ndata}(\XDaug^{\Nsamps})$.
	
	Due to \Cref{lemma:ConvergenceofCmtoCN}, the conditions required by \Cref{theorem:Shapiro} hold for every fixed $\Ndata$. Furthermore, since the inequality $\C_{\Ndata^*}^* < 1-\delta$ holds strictly, \Cref{alg:DataSelCtrl} returns a solution of size at most $\Ndata^*$ with probability $1$ for $\Nsamps$ large enough. As the samples drawn for each problem are i.i.d., we have
	\begin{align*}
	\begin{split}
		&\text{P}\Big(\lim \limits_{\Nsamps \rightarrow \infty }\C_{\Ndata}^{\Nsamps,*} = \C^{*},  \lim_{\Nsamps \rightarrow \infty }\lvert \XDaug^{\Nsamps} \rvert = \Ndata^* ,  \forall \ \Ndata \in \mathbb{N}_{\leq \Ndata^*} \Big) \\= &\prod \limits_{\Ndata=1}^{\Ndata^*} \text{P}\left(\lim \limits_{\Nsamps \rightarrow \infty }\C_{\Ndata}^{\Nsamps}(\XDaug^{\Nsamps,*}) = \C^{*},  \lim_{\Nsamps \rightarrow \infty }\lvert \XDaug^{\Nsamps} \rvert = \Ndata^* \right)=1
		\end{split}
	\end{align*}
	In particular, \Cref{theorem:MainResult} implies that, for $\Nsamps$ large enough, the difference between the value of the approximate optimal value $\C_{\Ndata}^{\Nsamps,*}$ and the exact optimal solution $\C^{*}$ can be made arbitrarily small.
	
\end{proof}

\section{Numerical Illustration}
\label{sec:simulation}

We illustrate the proposed approach with a system of the form given by \eqref{eq:SystemDynamics}, where $\g(\xaug)=(u_1, \ u_2)\transp$,
\begin{align*}
	\f(\xaug)=\begin{pmatrix}
	x_1+(\cos(2\pi x_1)-1)x_2 \\
	\frac{1}{1+\exp(-5x_1)-\frac{1}{2}+\cos(\pi x_2)}
	\end{pmatrix},
\end{align*}
and $\w_{\step} \sim\mathcal{N}(\bm{0},\text{diag}(0.01,0.01))$.
Due to its highly nonlinear dynamics, it is impossible to extrapolate the system's behavior from locally collected data. Hence, unless the regions of interest for each control task overlap, each control task requires different measurements to achieve good performance.

We assume to know that $\f(\cdot)$ depends exclusively on $\x$, hence we employ a GP that takes only the state $\x$ as input. Moreover, we employ a squared-exponential kernel $\kernel(\cdot,\cdot)$ for the GP, which is able to approximate a continuous function arbitrarily accurately on compact sets \cite{wahba1990spline}. We employ GP-based feedback linearizing control laws $\bm{u}^j(\x,t)=-\postmean_{\Ndata,\step}(\bm{x})+\bm{x}_{\text{ref}}^j(t)$ with $3$ different reference trajectories 
\begin{align}
\x_{\text{ref}}^1(t)&=\bm{0}\\
\x_{\text{ref}}^2(t)&=\begin{bmatrix}
\sin(2\pi t/50) & \cos(2\pi t/50)
\end{bmatrix}\transp\\
\x_{\text{ref}}^3(t)&=\begin{bmatrix}
2\sin(2\pi t/25)& \cos(2\pi t/100)
\end{bmatrix}\transp.
\end{align}

The GP used to compute the mean $\postmean_{\Ndata,\step}(\cdot)$ is identical to the one used to obtain the approximate optimal data set $\XDaug^{\Nsamps,*}$ Each control law is required to fulfill a single tracking performance requirement $h_{\step}^j(\bm{x}) \leq 0$, $j=1,2,3$ where
\begin{align}
\label{eq:toyproblemreq12}
h_{\step}^j(\bm{x})&=\|\bm{x}-\bm{x}_{\text{ref},j}(t)\|_2-\varphi(\step), \ j=1,2,\\
\label{eq:toyproblemreq3}
h_{\step}^3(\bm{x})&=|x_1|-5/2,
\end{align}
and $\varphi(\step)\coloneqq \max\{ 3\exp(-\step/5),0.1 \}$,
over a time horizon of $\Nhor = 100$ steps. We assume that the optimal data set is contained within $\StatSpAug^* = [-3,3]^2$, since the control objectives are restricted to this region. Furthermore, we are given $100$ prior measurements taken from random samples of the true system, which we use to train the GP. The number of samples used to obtain the approximate optimal data set $\XDaug^{\Nsamps,*}$ is set to $\Nsamps =100$, and the desired probability of constraint satisfaction is set to $1-\delta=0.01$.

In order to solve the approximate optimization problem, we search for a solution by minimizing the surrogate function
\begin{align*}
\frac{1}{\Nsamps}\sum \limits_{m}^{\Nsamps}\prod\limits_{j=1}^{\Nlaws} \prod\limits_{t=1}^{\Nhor} \hsc^j_{\step}(\xaugsamp^{j}_{\step,m})\bm{1}_{\mathbb{R}_{+}^{\Ncons}}\left(\hsc^j_{\step}(\xaugsamp^{m}_{n_{j,\step}}) \right),
\end{align*}
which enables us to employ gradient-based methods. 

We apply the DS-ML algorithm $10$ times using randomly sampled starting points $\x_0 \in \mathcal{U}([-3,3]^2)$, where $\mathcal{U}(\cdot)$ denotes a uniform distribution, and obtain an approximate optimal data set $\XDaug^{\Nsamps,*}$ after $\Ndata \in \{6, \cdots, 9\}$ iterations of \Cref{alg:DataSelCtrl}. The approximate probability of constraint violation as a function of $\Ndata$ is shown in \Cref{fig:p_N}. The prior system measurements, the desired trajectories, and an approximate optimal set $\XDaug^{\Nsamps,*}$ obtained after applying the DS-ML algorithm can be seen in \Cref{fig:dat_dist}. 

\begin{figure}
	\pgfplotsset{ 
		height =50\columnwidth /100, grid= major, 
		legend cell align = left, ticklabel style = {font=\scriptsize},
		every axis label/.append style={font=\small},
		legend style={
			draw=none, 
			text depth=0pt,
			at={(-0.03,1.3)},
			anchor=north west,
			legend columns=-1,
			column sep=0.5cm,
			%
			/tikz/every odd column/.append style={column sep=0cm},
	}	 }
	\center
	\def\file{plots/dat_dist.txt}
	\tikzsetnextfilename{dat_dist}
	\begin{tikzpicture}[trim axis right,trim axis left]
	\begin{axis}[view={0}{90},grid=none,enlargelimits=false, axis on top,
	xlabel={$x_1$}, ylabel={$x_2$}, 
	xmin=-2.5, xmax = 2.5, ymin = -1.2, ymax = 1.2,	
	y label style={yshift=-0.2cm},
	]
	\addlegendimage{empty legend}
	\addplot[only marks,mark size=1pt] table[x=X_1, y=X_2]{\file};
	\addplot[only marks,red,mark size=1pt] table[x=Xdat_1, y=Xdat_2]{\file};
	\addplot[green,dashed] table[x=Xref3_1, y=Xref3_2]{\file};
	\addplot[green] table[x=Xref2_1, y=Xref2_2]{\file};
	\addplot[green, only marks, mark=+] table[x=Xref1_1, y=Xref1_2]{\file};
	\legend{Legend:, Prior data,$\XDaug^{M,*}$,$\x_{\text{ref},j}(\step)$}
	\end{axis}
	\end{tikzpicture}
	\caption{Prior measurement data, reference trajectories $\x_{\text{ref},j}(\step)$, and approximate optimal measurement locations $\XDaug^{M,*}$ obtained with a single application of DS-ML algorithm using $\Nsamps=50$.}
	\label{fig:dat_dist}
\end{figure}
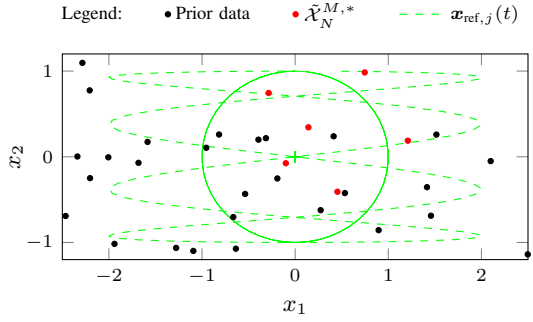

All approximate optimal sets $\XDaug^{\Nsamps,*}$ correspond roughly to points within the circle given by $\x_d^2(\step)$. This result is intuitive, since this is the region where the desired trajectories specified by \cref{eq:toyproblemreq12,eq:toyproblemreq3} overlap the most. Moreover, as can be seen in \Cref{fig:dat_dist}, the approximate optimal solution $\XDaug^{\Nsamps,*}$ regions that are both unexplored and of interest to the individual control tasks. However, since we employed a gradient-based solver, sub-optimal solutions are to be expected. This is also the case in \Cref{fig:dat_dist}, where some data points are close to already available prior data, i.e., a local minimum was found.

After every completion of the DS-ML algorithm, measurements of the true system at the approximate optimal set $\XDaug^{\Nsamps,*}$ are collected, and we carry out $100$ Monte Carlo simulations of the true system. This results in no constraint violation except for task $j=2$. However, constraint violations are small, as can be seen in \Cref{fig:const_real}, which indicates that the proposed method yielded a good approximate optimal data set $\XDaug^{\Nsamps,*}$.

\begin{figure}
\pgfplotsset{ compat = 1.13, 
	height =45\columnwidth /100, grid= major, }
	\center
	\def\file{plots/p_N.txt}
	\tikzsetnextfilename{p_N}
	\begin{tikzpicture}[trim axis right,trim axis left]
	\begin{axis}[grid=none,enlargelimits=false, axis on top,
	restrict x to domain*=0:25,
	restrict y to domain*=0:1,
	xlabel={Size of data set $\Ndata$}, ylabel={$\C_{\Ndata}^{\Nsamps}(\XDaug)$}, 
	xmin=0, xmax = 10, ymin = 0, ymax = 1.1,	
	y label style={yshift=-0.2cm},
	legend style={at={(1.02,0.5)},anchor=west},
	]
	\addplot[fill=black!20,on layer=axis background,draw=none] table[x = Nsul,y = p_stdul ]{\file};
	\addplot[black] table[x=Ns, y=p_m]{\file};
	\end{axis}
	\end{tikzpicture}
	\caption{Maximal approximate probability of constraint satisfaction $\C_{\Ndata}^{\Nsamps}(\XDaug^{\Nsamps})$ as a function of data set size $\Ndata$ for $50$ repetitions of DS-ML algorithm. Desired probability of constraint satisfaction $1-\delta$ is achieved after $\Ndata \in \left\{6,\cdots,9\right\}$ iterations of the DS-ML algorithm.}
	\label{fig:p_N}
\end{figure}
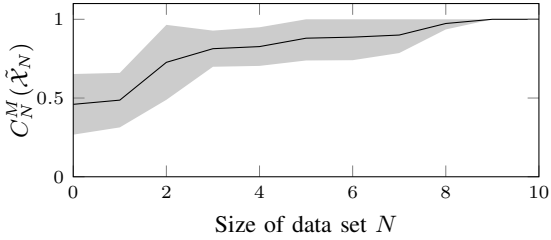

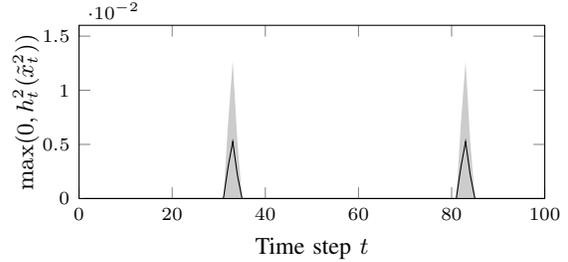
\begin{figure}
	
		\pgfplotsset{compat = 1.13, 
			height =45\columnwidth /100, grid= major,}
	\center
	\def\file{plots/const_real.txt}
	\tikzsetnextfilename{const_real}
	\begin{tikzpicture}[trim axis right,trim axis left]
	\begin{axis}[grid=none,enlargelimits=false, axis on top,
	xlabel={Time step $\step$}, ylabel={ $\max(0,\hsc_\step^2(\xaugsc_{\step}^2))$}, 
	xmin=0, xmax = 100, ymin = 0, ymax = 0.016,	
	y label style={yshift=-0.2cm},
	legend style={at={(1.02,0.5)},anchor=west},
	]
	\addplot[fill=black!20,on layer=axis background,draw=none] table[x = T_Horul,y = c_stdul ]{\file};
	\addplot[black] table[x=T_Hor, y=c_m]{\file};
	\end{axis}
	\end{tikzpicture}
	\caption{Constraint violations yielded by applying control law $\uinsc_{\step}^2(\cdot)$ to true system after data was collected at approximate optimal set $\XDaug^{\Nsamps,*}$ computed by DL-MS algorithm. }
	\label{fig:const_real}
\end{figure}

\addtolength{\textheight}{-3cm}   

\section{CONCLUSION AND FUTURE WORK}
\label{sec:conclusion}
This paper presents an algorithm to approximate the smallest training set required for successfully completing multiple tasks in learning-based 
control. We use a sample-based approximation that approximates the correct solution arbitrarily well with probability 1 as the number of 
samples increases. In a numerical simulation, the approximate optimal data sets computed with the proposed method are shown to yield adequate data sets for multiple tasks.

\bibliographystyle{IEEEtran}
\bibliography{../Literature/AllPhDReferences}

\end{document}